\documentclass{sig-alternate}

\usepackage{graphicx}
\usepackage{algorithmic}
\usepackage{amsfonts}
\usepackage{amsmath,amssymb,bbm}
\usepackage{graphicx}
\usepackage{mathrsfs}

\usepackage{multicol}
\usepackage{longtable}

\def\LM{{\mathrm{LM}}}
\def\LC{{\mathrm{LC}}}
\def\LT{{\mathrm{LT}}}

\def\S{{\mathrm{Spoly}}}

\def \lcm{{\rm lcm}}
\def\ri{\rangle}
\def\li{\langle}

\def\H{\mathcal{H}}
\def\I{\mathcal{I}}

\def\J{\mathcal{J}}

\def\mult#1#2#3{{#1}_{#2}(#3)}
\def\NN{\mathbbm{N}}
\def\TT{\mathbbm{T}}

\def\P{\mathcal{P}}

\def\X{\mathcal{X}}

\def\H{\mathcal{H}}
\def\I{\mathcal{I}}
\def\J{\mathcal{J}}

\def\kk{\mathbbm{k}}

\def\P{\mathcal{P}}

\def\R{\mathcal{R}}
\def \LT{{\rm LT}}
\def \LM{{\rm LM}}
\def \LC{{\rm LC}}
\def \HF{{\rm HF}}
\def \lcm{{\rm lcm}}
\def \HP{{\rm HP}}
\def \HS{{\rm HS}}

\def \hilb{{\rm hilb}}

\def\ri{\rangle}
\def\li{\langle}

\DeclareMathOperator{\cls}{cls}

\DeclareMathOperator{\depth}{depth}
\DeclareMathOperator{\gin}{gin}
\DeclareMathOperator{\lt}{lt}

\DeclareMathOperator{\reg}{reg}

\newtheorem{theorem}{Theorem}[section]
\newtheorem{proposition}[theorem]{Proposition}
\newtheorem{definition}[theorem]{Definition}
\newtheorem{corollary}[theorem]{Corollary}
\newtheorem{lemma}[theorem]{Lemma}
\newtheorem{example}[theorem]{Example}

\newtheorem{remark}[theorem]{Remark}
\newtheorem{conjecture}[theorem]{Conjecture}

\begin{document}
%
% --- Author Metadata here ---
\conferenceinfo{ISSAC}{'17  Kaiserslautern,  Germany}

\conferenceinfo{ISSAC'17,} {July 25-28, 2017, Kaiserslautern,  Germany.}

\CopyrightYear{2012}

%\crdata{978-1-59593-904-3/08/07}

\title{Dimension-Dependent Upper Bounds for Gr\"obner Bases}
\numberofauthors{2}

\author{
\alignauthor
Amir Hashemi\\
       \affaddr{Department of Mathematical Sciences}\\  \affaddr{Isfahan University of Technology}\\ \affaddr{Isfahan, 84156-83111, Iran}\\
       \affaddr{School of Mathematics, Institute for Research in Fundamental Sciences (IPM)}\\ \affaddr{Tehran, 19395-5746, Iran}
       \email{Amir.Hashemi@cc.iut.ac.ir}
       \alignauthor
Werner M. Seiler\\
\affaddr{Institut f\"ur Mathematik, Universit\"at Kassel, Heinrich-Plett-Stra\ss e 40, 34132 Kassel, Germany}\\
\email{seiler@mathematik.uni-kassel.de}
}\maketitle

\begin{abstract}
  We improve certain degree bounds for Gr\"obner bases of polynomial ideals
  in generic position. We work exclusively in deterministically verifiable
  and achievable generic positions of a combinatorial nature, namely either
  strongly stable position or quasi stable position. Furthermore, we
  exhibit new dimension- (and depth-)dependent upper bounds for the
  Castelnuovo-Mumford regularity and the degrees of the elements of the
  reduced Gr\"obner basis (w.r.t.\ the degree reverse lexicographical
  ordering) of a homogeneous ideal in these positions.
\end{abstract}

%\category{I.1.2}{Symbolic and Algebraic Manipulation}{Algorithms}

%\terms{Algorithms, Theory, Experimentation}


%%%%%%%%%%%%%%%%%%%%%%%%%%%%%%%%%%%%%%%%%%%
\section{Introduction}
\label{sec:1}

Gr\"obner bases, introduced by Bruno Buchberger in his Ph.D. thesis (see
e.g. \cite{Bruno1,Bruno2}), have become a powerful tool for constructive
problems in polynomial ideal theory and related domains. For practical
applications, in particular, the implementation in computer algebra
systems, it is important to establish upper bounds for the complexity of
determining a Gr\"obner basis for a given homogeneous polynomial ideal. Using Lazard's
algorithm \cite{Daniel83}, a good measure to estimate such a bound, is an
upper bound for the degree of the intermediate polynomials during the
Gr\"obner basis computation. If the input ideal is {\em not} homogeneous,
the maximal degree of the output Gr\"obner basis is {\em not} sufficient
for this estimation. On the other hand, M\"oller and Mora \cite{Moller}
showed that to discuss degree bounds for Gr\"obner bases, one can restrict
to homogeneous ideals. Thus upper bounds for the degrees of the elements of
Gr\"obner bases of homogeneous ideals, allow us to estimate the complexity
of computing Gr\"obner bases in general.

Let us review some of the existing results in this direction. Let $\P$ be
the polynomial ring $\kk[x_1,\ldots, x_n]$ where $\kk$ is of characteristic
zero and $\I\subset \P$ be an ideal generated by \emph{homogeneous}
polynomials of degree at most $d$ with $\dim(\I)=D$. The first doubly
exponential upper bounds were proven by Bayer, M\"oller, Mora and Giusti,
see \cite[Chapter 38]{Mora} for a comprehensive review of this topic. Based
on results due to Bayer \cite{bayer} and Galligo \cite{Gal1,Gal2}, M\"oller
and Mora \cite{Moller} provided the upper bound $(2d)^{(2n+2)^{n+1}}$ for
any Gr\"obner basis of $\I$. They also proved that this doubly exponential
behavior cannot be improved. Simultaneously, Giusti \cite{Giusti1} showed
the upper bound $(2d)^{2^{n-1}}$ for the degree of the reduced Gr\"obner
basis (w.r.t.\ the degree reverse lexicographic order) of $\I$ when the
ideal is in {\em generic position}. Then, using a self-contained and
constructive combinatorial argument, Dub\'e \cite{Dube} proved the so far
sharpest degree bound $2( d^2/2 + d)^{2^{n-1}}\sim 2d^{2^n}$.

In 2005, Caviglia and Sbarra \cite{Caviglia} studied upper bounds for the
Castelnuovo-Mumford regularity of homogeneous ideals.  Analyzing Giusti's
proof, they gave a simple proof of the upper bound $(2d)^{2^{n-2}}$ for the
degree reverse lexicographic Gr\"obner basis of an ideal $\I$ in generic
position (they also showed that this bound holds independent of the
characteristic of $\kk$). Finally, Mayr and Ritscher \cite{MR}, by
following the tracks of Dub\'e \cite{Dube}, obtained the
dimension-dependent upper bound $2(1/2d^{n-D}+d)^{2^{D-1}}$ for any reduced
Gr\"obner basis of $\I$. It is worth while remarking that there are also
\emph{lower} bounds for the complexity: $d^{2^m}$ with $m=n/10-O(1)$ from
the work of Mayr and Meyer \cite{MM} and $d^{2^m}$ where $m\sim n/2$ due to
Yap \cite{Yap}.

In this article, we will first improve Giusti's bound by showing that if
$\I$ is in strongly stable position and $D>1$, then $2d^{(n-D)2^{D-1}}$ is
a simultaneous upper bound for the Castelnuovo-Mumford regularity of $\I$
and for the maximal degree of the elements of the Gr\"obner basis of $\I$
(with respect to the degree reverse lexicographic order). Furthermore, we
will sharpen the bound of Caviglia-Sbarra to
$(d^{n-D}+(n-D)(d-1))^{2^{D-1}}$. We will see that neither of these bounds
is always greater than the other. Finally, we will show that, if $\I$ is in
quasi stable position and $D\le 1$, Giusti's bound may be replaced by
$nd-n+1$ (this result was already obtained by Lazard \cite{Daniel83} when
the ideal is in generic position).  In the recent work \cite{hashemi3}, we
showed how many variants of stable positions -- including quasi stable and
strongly stable position -- can be achieved via linear coordinate
transformations constructed with a deterministic algorithm.

The article is organized as follows. In the next section, we give basic
notations and definitions.  In Sections $4,5$ and $6$ we improve the degree
bounds provided by Giusti, Caviglia-Sbarra and Lazard, respectively.

%%%%%%%%%%%%%%%%%%%%%%
\section{Preliminaries}
\label{sec:2}

Throughout this article, we keep the following notations. Let
$\P=\kk[x_{1},\dots,x_{n}]$ be the polynomial ring (where $\kk$ is of
characteristic zero). A power product of the variables $x_{1},\dots,x_{n}$
is called {\em term} and $\TT$ denotes the monoid of all terms in $\P$. We
consider non-zero homogeneous polynomials $f_1,\ldots ,f_k\in \P$ and the
ideal $\I=\li f_1,\ldots ,f_k \ri$ generated by them.  We assume that $f_i$
is of degree $d_i$ and that the numbering is such that
$d_1\ge d_2\ge \cdots\ge d_k>0$.  We also set $d=d_1$.  Furthermore, we
denote by $\R=\P/\I$ the corresponding factor ring and by $D$ its
dimension. Finally, we use throughout the degree reverse lexicographic
order with $x_n\prec \cdots \prec x_1$.

 The {\em leading term} of a polynomial $f\in \P$, denoted by ${\LT}(f)$, is
the greatest term (with respect to $\prec$) appearing in $f$ and its
coefficient is the {\em leading coefficient} of $f$ and we denote it by
${\LC}(f)$. The {\em leading monomial} of $f$ is the product
${\LM}(f)={\LC}(f){\LT}(f)$. The {\em leading ideal} of $\I$ is defined as
${\LT}(\I) = \langle {\LT}(f) \ | \ f \in {\I}\rangle$. For the finite set
$F=\{f_1,\ldots ,f_k\}\subset {\P}$, $\LT(F)$ denotes the set
$\{\LT(f_1),\ldots ,\LT(f_k)\}$. A finite subset
$G \subset {\I}$ is called a {\em Gr\"obner basis} of
$\I$ w.r.t.\ $\prec$, if ${\LT}(\I) = \langle \LT(G) \rangle$.  We refer to
\cite{Becker} for more details on Gr\"obner bases.

Given a graded $\P$-module $X$ and a positive integer $s$, we denote by
$X_{s}$ the set of all homogeneous elements of $X$ of degree $s$. To define
the Hilbert regularity of an ideal, recall that the {\em Hilbert function}
of $\I$ is defined by $\HF_{\I}(t)=\dim_{\kk}(\R_t)$; the dimension of
$\R_t$ as a $\kk$-linear space. From a certain degree on, this function of
$t$ is equal to a polynomial in $t$, called {\em Hilbert polynomial}, and
denoted by $\HP_{\I}$ (see \cite{little} for more details on this
topic). The {\em Hilbert regularity} of $\I$ is
$\hilb(\I)=\min \{m \ \arrowvert \ \forall t \geq m , \
\HF_{\I}(t)=\HP_{\I}(t)\}$. Finally, recall that the {\em Hilbert series}
of $\I$ is the power series
${\HS_{\I}}(t)=\sum_{s=0}^{\infty}{{\HF_{\I}}(s)t^{s}}$.

\begin{proposition}\label{hilbI}
  There exists a univariate polynomial $p(t)$ with $p(1)\ne 0$ such that
  ${\HS_{\I}}(t)=p(t)/(1-t)^D$. Furthermore,
  $\hilb(\I)=\max\{0,\deg(p)-d+1\}$.
\end{proposition}

For a proof of this result, we refer to \cite[Thm. 7, page 130]{Ralf}. It
follows immediately from Macaulay's theorem that the Hilbert function of
$\I$ is the same as that of $\LT(\I)$ and this provides an effective method
to compute it using Gr\"obner bases, see e.g.\ \cite{Singular}.

Let us state some auxiliary results on regular sequences. Recall that a
sequence of polynomials $f_1,\ldots ,f_k\in \P$ is called {\it regular} if
$f_i$ is a non-zero divisor on the ring $\P/\li f_1,\ldots ,f_{i-1}\ri$ for
$i=2,\ldots ,k$. This is equivalent to the condition that $f_i$ does not
belong to any associated prime of $\li f_1,\ldots ,f_{i-1}\ri$. It can be
shown that the Hilbert series of a regular sequence $f_1,\ldots, f_k$ is
equal to $\prod_{i=1}^k(1-t^{d_i})/(1-t^n)$, see e.g.\ \cite{Monique}.  The
converse of this result is also true, see \cite[Exercise 7, page
137]{Ralf}. In addition, these conditions are equivalent to the statement
that $D=n-k$. 

\begin{lemma}(\cite[Prop.\ 4.1, page 108]{Monique})\label{Sec3:Lem1} There
  exist homogeneous polynomials $g_1,\ldots,g_{n-D}\in \P$ such that the
  following conditions hold:
  \begin{itemize}
  \item[$(1)$] $\deg(g_i)=d_i$ for each $i$,
  \item[$(2)$] $g_i\equiv \lambda_i f_i \mod \li f_{i+1},\ldots, f_{k}\ri$
    for some $0 \ne \lambda_i \in \kk$ for $i=1,\ldots,n-D$,
  \item[$(3)$] $g_1,\ldots,g_{n-D}$ is regular sequence in $\P$.
  \end{itemize}
\end{lemma}

\begin{definition}
  The {\em depth} of the homogeneous ideal $\I$ is defined as the maximal
  integer $\lambda$ such that there exists a regular sequence of linear
  forms $y_1,\ldots, y_\lambda$ on $\P/\I$.
\end{definition}

\begin{definition}\label{defbayer}
  The homogeneous ideal $\I$ is {\em $m$-regular}, if its minimal graded
  free resolution is of the form
  \begin{multline*}
    0\longrightarrow \bigoplus _{j}\P(e_{rj})\longrightarrow \cdots\\
    \cdots \longrightarrow \bigoplus _{j}\P(e_{1j})\longrightarrow
    \bigoplus_{j}\P(e_{0j})\longrightarrow \I\longrightarrow 0
  \end{multline*}
  with $e_{ij}-i \leq m$ for each $i,j$. The {\em Castelnuovo-Mumford
    regularity} of $\I$ is the smallest $m$ such that $\I$ is $m$-regular;
  we denote it by $\reg(\I)$.
\end{definition}

For more details on the regularity, we refer to
\cite{Mumford,Eisenbud_Goto,bayer_stillman,bermejo}. It is well-known that
in generic coordinates $\reg(\I)$ is an upper bound for the degree of the
Gr\"obner basis w.r.t.\ the degree reverse lexicographic order. This upper
bound is sharp, if the characteristic of $\kk$ is zero (see
\cite{bayer_stillman}).  A good measure to estimate the complexity of the
computation of the Gr\"obner basis of $\I$ is the maximal degree of the
polynomials which appear in this computation (see
\cite{Lazard81,Daniel83,Giusti1}).

\begin{definition}
  We denote by $\deg (\I, \prec)$ the maximal degree of the elements of the
  reduced Gr\"obner basis of the non-zero homogeneous ideal $\I$ w.r.t.\ the
  term order $\prec$.
\end{definition}

\begin{theorem}(\cite[Prop. 4.8, page 117]{Monique}) \label{maincor} If
  $\I$ is zero-dimensional, then $\deg(\I,\prec)\le d_1+\cdots +d_n-n+1$.
\end{theorem}

We conclude this section with a brief review of the theory of Pommaret
bases.  Suppose that $f\in \P$ and $\LT(f)=x^\alpha$ with
$\alpha=(\alpha_1,\ldots ,\alpha_n)$.  We call
$\max{\{i\mid\alpha_{i}\neq0\}}$ the \emph{class} of $f$, denoted by
$\cls({f})$. Then the \emph{multiplicative variables} of $f$ are
$\mult{\X}{P}{f}=\{x_{\cls({f})},\ldots ,x_n\}$.  Furthermore, $x^{\beta}$
is a \emph{Pommaret divisor} of $x^{\alpha}$, written
$x^{\beta}\mid_P x^{\alpha}$, if $x^{\beta}\mid x^{\alpha}$ and
$x^{\alpha-\beta}\in\kk[x_{\cls({f})},\ldots ,x_n]$.

\begin{definition}\label{def:pombas}
  Let $\H\subset \I$ be a finite set such that no leading term of an
  element of $\H$ is a Pommaret divisor of the leading term of another
  element. Then $\H$ is called a \emph{Pommaret basis} of $\I$ for $\prec$,
  if
  \begin{equation}\label{eq:pombas}
    \I=\bigoplus_{h\in\H}\kk[\mult{\X}{P}{h}]\cdot h.
  \end{equation}
\end{definition}

One can easily show that any Pommaret basis is a (generally non-reduced)
Gr\"obner basis of the ideal it generates.  The main difference between
Gr\"obner and Pommaret bases consists of the fact that by (\ref{eq:pombas})
any polynomial $f\in \I$ has a {\em unique} involutive standard
representation. If an ideal $\I$ possesses a Pommaret basis $\H$, then
$\reg(\I)$ equals the maximal degree of an element of $\H$, cf.\ \cite[Thm.
9.2]{wms:comb2}. The main drawback of Pommaret bases is however that they
do not always exist. Indeed, a given ideal possesses a finite Pommaret
basis, if and only if the ideal is in {\em quasi stable position} -- see
\cite[Prop 4.4]{wms:comb2}.

\begin{definition}
  A monomial ideal $\J$ in $\P$ is called \emph{quasi stable}, if for any
  term $m \in \J$ and all integers $i, j,s$ with $1 \le j < i \le n$ and
  $s>0$ such that $x_i^s\mid m$, there exists an exponent $t\ge 0$ such
  that $x_j^tm/x_i^s\in \J$. A homogeneous ideal $\I$ is in {\em quasi
    stable position}, if $\LT(\I)$ is quasi stable.
\end{definition}

In the sequel, we will use the following notations: given an ideal $\I$ in
quasi stable position, we write $\H = \{h_1 ,\ldots , h_s \}$ for its
Pommaret basis. Furthermore, for each $i$ we set $m_i = \LT(h_i )$ and it
is then easy to see that $\{m_1 , \ldots , m_s \}$ forms a Pommaret basis
of $\LT(\I)$.

\begin{remark}\label{rem} 
  Since any linear change of variables is a $\kk$-linear automorphism of
  $\P$ preserving the degree, it follows trivially that the dimensions over
  $\kk$ of the homogeneous components of a homogeneous ideal $\I$ or of its
  factor ring $\R$ remain invariant. Hence the Hilbert function and
  therefore also the Hilbert series, the Hilbert polynomial and the Hilbert
  regularity of $\I$ do not change.  The same is obviously true for the
  Castelnuovo-Mumford regularity. In addition, due to the special form of
  the Hilbert series of the ideal generated by a regular sequence, we
  conclude that any regular sequence remains regular after a linear change
  of variables and hence the depth is invariant, too.  Finally, we note
  that almost all linear changes of variables transform a given homogeneous
  ideal into quasi stable position (which is thus a generic position)
  \cite{{wms:comb2}}. It follows that to study any of the mentioned
  invariants of $\I$, w.l.o.g.\ we may assume that $\I$ is in quasi stable
  position.
\end{remark}

%%%%%%%%%%%%%%%%%%%%%%%%%
\section{Improving Giusti's upper bound}
\label{sec:4}

In 1984, Giusti \cite{Giusti1} established the upper bound $(2d)^{2^{n-1}}$
for $\deg(\I,\prec)$ in the case that the coordinates are in generic
position. The key point of Giusti's proof is the use of the combinatorial
structure of the generic initial ideal in characteristic zero. Later on,
Mora \cite[Ch.\ 38]{Mora}, by a deeper analysis of Giusti's proof, improved
this bound to $(d+1)^{(n-D)2^{D-\lambda}}$ where $\lambda$ is the depth of
$\I$.  In this section, we improve Mora's bound by following his general
approach and correcting some flaws in his method. Our presentation seems to
be simpler than the ones by Mora and Giusti.

We first note that for a given ideal in quasi stable position, we are able
to reduce the number of variables by the depth of the ideal to obtain a
sharper bound for $\deg(\I,\prec)$. A novel proof {\em \`a la Pommaret} of
this result is given below.

\begin{proposition}\label{propdepth}
  Suppose that $U(n,d,D)$ is a function depending in $n,d$ and $D$ so that
  $\deg(\I,\prec)\le U(n,d,D)$ for any ideal $\I$ which is in quasi stable
  position and is generated by homogeneous polynomials of degree at most
  $d$ in $n$ variables. Then, $\deg(\I,\prec)\le U(n-\lambda,d,D-\lambda)$
  where $\depth(\I)=\lambda$.
\end{proposition}

\begin{proof}
  Let $t$ be the maximal class of the elements in $\H$. It is shown in
  \cite[Prop 2.20]{wms:comb2} that in quasi-stable position the variables
  $x_{t+1},\ldots, x_n$ define a regular sequence on $\R$ and that thus
  $\lambda=n-t$ (note that this reference distinguishes between
  $\depth(\I)$ and $\depth(\R)$ with the two related by
  $\depth(\R)=\depth(\I)-1$; what we call here $\depth(\I)$ corresponds to
  $\depth(\R)$ in \cite{wms:comb2}). By definition of $t$, no leading term
  of an element of $\H$ is divisible by any of these variables. Thus
  $\tilde{\H}=\H|_{x_{t+1}=\cdots=x_n=0}$ is the Pommaret basis of the
  ideal $\tilde{\I}=\I|_{x_{t+1}=\cdots=x_n=0}$ in
  $\kk[x_{1},\ldots,x_{t}]$ and hence
  $\deg(\I,\prec)=\deg(\tilde{\I},\prec)$. This entails our claim.  $\qed$
\end{proof}

\begin{corollary}
  As a similar statement to Prop.\ \ref{propdepth}, suppose that $R(n,d,D)$
  is a function depending in $n,d$ and $D$ such that
  $\reg(\I)\le R(n,d,D)$. Then, $\reg(\I)\le R(n-\lambda,d,D-\lambda)$.
\end{corollary}

\begin{proof}
  This claim follows by the same argument as in the proof of Prop.\
  \ref{propdepth} and using the facts that for each $f$ in the Pommaret
  bases $\H$ the corresponding element $\tilde{f}\in \tilde{\H}$ has the
  same degree as $f$ and in quasi stable position
  $\reg(\I)=\reg(\tilde{\I})$ is given by the maximal degree of the
  elements of $\H$ and $\tilde{\H}$. $\qed$
\end{proof}

To state the refined version of Giusti's bound, we need to recall the {\em
  crystallisation principle}.  Let $A=(a_{ij})\in \mathrm{GL}(n,\kk)$
be an $n\times n$ invertible matrix. By $A.\I$ we mean the ideal generated
by the polynomials $A.f$ with $f\in \I$ where
$A.f=f(\sum_{i=1}^{n}{a_{i1}x_i},\ldots,\sum_{i=1}^{n}{a_{in}x_i})$. The
following fundamental theorem is due to Galligo \cite{Gal1}.

\begin{theorem}\label{gal}
  There exists a non-empty Zariski open subset
  $\mathcal{U}\subset \mathrm{GL}(n,\kk)$ such that $\LT(A.\I)=\LT(A'.\I)$
  for all matrices $A,A'\in \mathcal{U}$.
\end{theorem}

\begin{definition}
  The monomial ideal $\LT(A.\I)$ with $A\in \mathcal{U}$ and $\mathcal{U}$
  as given in Theorem \ref{gal} is called {\em the generic initial ideal}
  of $\I$ (w.r.t.\ $\prec$) and is denoted by $\gin({\I})$.
\end{definition}

Suppose that $\I=\langle f_1,\ldots, f_k\rangle$ and that for some
$s\in\NN$ we have $\deg(f_i)\le s$ for all $i$ and $\gin(\I)$ has no
minimal generator in degree $s+1$. Then, the crystallisation principle (CP)
states that for each $m$ in the generating set of $\gin(\I)$ we have
$\deg(m)\le s$, see \cite[Prop 2.28]{Green}. Note that this principle holds
only in characteristic zero and it has been proven only for generic initial
ideals and for lexicographic ideals (see \cite[Thm. 3.8]{Green}).

Giusti's proof consists in applying this property along with an induction
on the number of variables. One crucial fact in this direction is that CP
also holds for a generic initial ideal modulo the last variable. Below, we
will show that both properties remain true for arbitrary {\em strongly
  stable} ideals.

\begin{definition}
  A monomial ideal $\J$ is called \emph{strongly stable}, if for any term
  $m \in \J$ we have $x_jm/x_i\in \J$ for all $i$ and $j$ such that $j<i$ and
  $x_i$ divides $m$. A homogeneous ideal $\I$ is in {\em strongly stable
    position}, if $\LT(\I)$ is strongly stable.
\end{definition}

\begin{proposition}\label{popstrong}
  Let $\I$ be in strongly stable position. Then, CP holds for $\LT(\I)$. 
\end{proposition}

\begin{proof}
  The following arguments are inspired by \cite[page 728]{Mora}.  Let us
  consider an integer $s\ge d$.  Suppose that we are computing a Gr\"obner
  basis of $\I$ using Buchberger's algorithm and by applying the normal
  strategy. In addition, assume that we have already computed the set
  $G=\{g_1,\ldots g_t\}$ up to degree $s$ (this set will be enlarged to a
  Gr\"obner basis of $\I$), and there is no new polynomial of degree $s+1$
  to be added into $G$. Note that we have chosen $s\ge d$ to be sure that
  $G$ generates $\I$. To prove the assertion, it suffices to show that $G$
  is a G\"obner basis of $\I$.

  We introduce the set $M_{s}=\langle\LT(G)\rangle_{s}\cap\TT$.  We now
  claim that for each pair of terms
  $x^\alpha=x_1^{\alpha_1}\cdots x_n^{\alpha_n}\neq
  x^\beta=x_1^{\beta_1}\cdots x_n^{\beta_n}$ in it either
  $\deg(\lcm(x^\alpha,x^\beta))=s+1$ or there exists a further term
  $x^\gamma\in M_s\setminus\{x^\alpha,x^\beta\}$ such that
  \begin{itemize}
  \item $x^\gamma\mid \lcm(x^\alpha,x^\beta)$,
  \item $\deg(\lcm(x^\gamma,x^\alpha))< \deg(\lcm(x^\alpha,x^\beta))$,
  \item $\deg(\lcm(x^\gamma,x^\beta))< \deg(\lcm(x^\alpha,x^\beta))$.
  \end{itemize}
  If this claim is true, then Buchberger's second criterion implies that it
  suffices to consider those pairs $\{g_i,g_j\}$ with
  $\deg(\lcm(\LT(g_i),\LT(g_j)))=s+1$. If for each such pair the
  corresponding S-polynomial reduces to zero, then $G$ is a Gr\"obner basis
  and we are done. Otherwise, there exists a new generator of degree $s+1$
  contradicting the made assumptions.

  For proving the made claim, it suffices to show that, if
  $\deg(\lcm(x^\alpha,x^\beta))> s+1$, then there exists a term
  $x^\gamma\in M_s\setminus\{x^\alpha,x^\beta\}$ satisfying the above
  conditions. Let $j$ be an integer such that $\alpha_j\ne \beta_j$ and
  $\alpha_{j+1}= \beta_{j+1},\ldots,\alpha_{n}= \beta_{n}$. W.l.o.g., we
  may assume that $\alpha_j>\beta_j$. Since $x^\alpha$ and $x^\beta$ have
  the same degree, there is an index $i<j$ such that $\beta_i>\alpha_i$.
  The strongly stable position of $\I$ implies that $M_s$ is a strongly
  stable set. Therefore the term $x^\gamma=x_ix^\alpha/x_j$ satisfies
  $x^\gamma\in M_s\setminus\{x^\alpha,x^\beta\}$ and
  $x^\gamma\mid \lcm(x^\alpha,x^\beta)$. Furthermore,
  $\deg(\lcm(x^\gamma,x^\alpha))=s+1<\deg(\lcm(x^\alpha,x^\beta))$ and
  $\deg(\lcm(x^\gamma,x^\beta))= \deg(\lcm(x^\alpha,x^\beta))-1$. $\qed$
\end{proof}

\begin{example}\label{ex:green}
  One should note that strong stability of the leading term ideal does not imply
  that it is the generic initial ideal, as the following example due to
  Green \cite{Green} shows:
  $\I=\langle x_1x_3,x_1x_2+x_2^2,x_1^2\rangle\subset
  \kk[x_1,x_2,x_3]$. Its leading term ideal
  $\LT(\I)=\langle x_1x_3,x_1x_2,x_1^2,x_2^2x_3,x_2^3\rangle$ is strongly
  stable, but we find
  $\gin(\I)=\langle x_2^2,x_1x_2,x_1^2,x_1x_3^2\rangle\neq
  \LT(\I)$. Nevertheless, it is clear that both $\LT(\I)$ and $\gin(\I)$
  satisfy CP.
\end{example}

As a consequence of the proof of this proposition, we can infer a
generalization of CP.

\begin{corollary}
  Suppose we know in advance that $\I$ is in strongly stable position. Let
  us fix an integer $t$ (not necessarily greater than $d$). Suppose that we
  are computing a Gr\"obner basis for $\I$ using Buchberger's algorithm and
  applying the normal strategy. Assume that we have treated all
  S-polynomials of degree at most $t$ and $G_t$ is the set of all
  polynomials computed so far. If all S-polynomials of degree $t+1$ reduce
  to zero, then any critical pair $\{f,g\}$ with
  $\max\{\deg(f),\deg(g)\}\le t$ is superfluous. In particular, $G_t$ is a
  Gr\"obner basis for $\langle\I_{\leq t}\rangle$.
\end{corollary}

In the sequel, for an index $i$ we denote by $\I_i$ the ideal
$\I|_{x_i=\cdots= x_n=0}\subset \kk[x_1,\ldots,x_{i-1}]$.  Since we assume
that $\prec$ is the degree reverse lexicographic term order, strongly
stable position of $\I$ entails that $\I_{i}$ is in strongly stable
position, too, for any index $i$. The essence of Giusti's approach consists
of finding, by repeated evaluation, relations between $\deg(\I,\prec)$ and
$\deg(\I_i,\prec)$ for $i=n,\ldots, n-D+1$. For this purpose, we introduce
some further notations for an ideal $\I$ in strongly stable position. We
denote by $N(\I)$ the set of all terms $m \notin \LT(\I)$.  If
$\dim(\I)=0$, then we define $F(\I)=N(\I)$. Otherwise we set
$F(\I)=\{\tau x_n^a \in N(\I)\ | \ \tau \in F(\I_n) \text{ and }
\deg(\tau x_n^a)< \deg(\I,\prec)\}$.  Since $\I$ is in strongly stable
position, $N(\I)$ is strongly stable for the reverse ordering of the
variables. More precisely, if $x^{\alpha}\in N(\I)$ with $\alpha_{i}>0$,
then we claim that $x_{j}x^{\alpha}/x_{i}\in N(\I)$ for any $j>i$. Indeed,
otherwise it belonged to $\LT(\I)$ and thus -- since $\LT(\I)$ is strongly
stable -- $x^{\alpha}\in \LT(\I)$ which is a contradiction.

\begin{lemma}\label{Mora}
  Suppose that $\I$ is in strongly stable position. Then the following
  statements hold.
  \begin{itemize}
  \item[$(a)$] $\deg(\I,\prec)\le \max\{d,\deg(\I_n,\prec)\}+\#F(\I_n)$,
  \item[$(b)$] $\#F(\I)\le \bigl(\max\{d,\#F(\I_n)\}\bigr)^2$.
  \end{itemize}
  (Here $\#X$ denotes the cardinality of a finite set $X$.)
\end{lemma}

\begin{proof}
  $(a)$ Let $G$ be the reduced Gr\"obner basis of $\I$ for $\prec$.
  Because of our use of the degree reverse lexicographic term order, we
  easily see that $G|_{x_n=0}$ is the reduced Gr\"obner basis of $\I_{n}$
  for $\prec$. Let $G'\subset G$ be the subset of all polynomials in $G$ of
  maximal degree.  We distinguish two cases. If
  $\LT(G')\cap \kk[x_1,\ldots ,x_{n-1}]\ne \emptyset$, then obviously
  $\deg{(\I,\prec)} = \deg(\I_n,\prec)$ and the assertion is
  proved. 

  Otherwise, CP (applicable by Prop.~\ref{popstrong}) implies that for each
  degree $\max{\{d,\deg(\I_n,\prec)\}}< i\le \deg(\I,\prec)$ there exists a
  polynomial $g_i\in G$ with $\deg(g_i)=i$ (note that if $\deg(\I,\prec)=d$
  then $(a)$ holds and we are done).  Thus, we can write $\LT(g_i)$ in the
  form $x_n^{a_i}\tau_i$ with $a_i>0$ and
  $\tau_i\in\kk[x_1,\ldots ,x_{n-1}]$. We claim that $\tau_i\in
  F(\I_n)$. Writing
  $\tau_i=x_{i_1}^{\alpha_{i_1}}\cdots x_{i_k}^{\alpha_{i_k}}$ where
  $i_1<\cdots <i_k$, we may conclude by the assumed reducedness of $G$ that
  $\tau_{i}\notin \LT(\I)$ and by the strong stability of $\LT(\I)$ that
  $x_{i_1}^{\alpha_{i_1}}\cdots x_{i_k}^{\alpha_{i_k}+a_i}\in
  \LT(\I)$. Hence there exists an integer $a>0$ such that
  $x_{i_1}^{\alpha_{i_1}}\cdots x_{i_k}^{\alpha_{i_k}+a-1}\notin \LT(\I)$
  and $x_{i_1}^{\alpha_{i_1}}\cdots x_{i_k}^{\alpha_{i_k}+a}\in
  \LT(\I)$. It follows that there exists a generator $g\in G\cap\I_{n}$
  such that its leading term
  $\LT(g)=x_{i_1}^{\beta_{i_1}}\cdots x_{i_k}^{\beta_{i_k}}$ divides the
  latter term.  We must have $\beta_\ell\le \alpha_\ell$ for each
  $\ell<i_k$ and $\beta_{i_k}=\alpha_{i_k}+a$ by definition of
  $a$. Furthermore, the strong stability of $\LT(\I)$ implies that
  $\deg(g)> \deg(\tau_i)$, as otherwise another generator $g'\in G$
  existed with $\LT(g')\mid\tau_{i}$. Thus
  $\deg(\tau_i)<\deg(\I_{n},\prec)$.  If we write
  $\tau_{i}=\bar{\tau}_{i}x_{i_{k}}^{\alpha_{i_{k}}}$, then there only
  remains to show that $\bar{\tau}_{i}\in F(\I_{i_{k}})$, as
  $\tau_{i}\in N(\I_{n})$ is a trivial consequence of $\tau_{i}\in N(\I)$.
  If $\dim(\I_{i_{k}})=0$, this follows immediately from
  $F(\I_{i_{k}})=N(\I_{i_{k}})$.  Otherwise we repeated the same arguments
  as above.

  Thus for each $i$ with $\max\{d,\deg(\I_n,\prec)\}< i\le \deg(\I,\prec)$
  there exists a generator $g_i\in G$ such that $\LT(g_i)=x_n^{a_i}\tau_i$
  and $\tau_i\in F(\I_n)$.  Since $G$ is reduced, the terms $\tau_i$ are
  pairwise different. Hence
  $\deg(\I,\prec)-\max{\{d,\deg(\I_n,\prec)\}} \le \#F(\I_n)$ and this
  proves $(a)$.
  % Indeed, on the basis of the above discussion, we
  % can conclude that if $\LM(g_i)=x_n^{a_i}\tau_i$ with
  % $\tau_i=x_{i_1}^{\alpha_{i_1}}\cdots x_{i_k}^{\alpha_{i_k}}$ then for
  % each $j<k$ we have
  % $\tau_j:=x_{i_1}^{\alpha_{i_1}}\cdots x_{i_j}^{\alpha_{i_j}}\in
  % N(\I_{i_j+1})\subset \kk[x_1,\ldots ,x_{i_j}]$ and
  % $\deg(\tau_i)\le \deg(\I_{i_j+1},\prec)$.  

  To show $(b)$, we introduce for each degree $\delta\in\NN$ the subset
  $F_{\delta}(\I)=\{x_n^\delta\tau \ | \ x_n^\delta\tau \in F(\I) \}$.  By
  definition, $x_n^\delta\tau\in F_{\delta}(\I)$ implies $\tau \in F(\I_n)$
  and thus $\#F_{\delta}(\I)\le \#F(\I_n)$. Since we used in the proof of
  $(a)$ CP, the claims proven there are true only for polynomials of degree
  at least $d$. Thus in the sequel we shall replace $\#F(\I_n)$ by
  $\max\{d,\#F(\I_n)\}$. We observe that the maximal $\delta$ such that
  $x_n^\delta\tau\in F(\I)$ is $\max\{d,\#F(\I_n)\}$ and thus
  \begin{displaymath}
    \#F(\I)\le\sum_{\delta=0}^{\max\{d,\#F(\I_n)\}-1}{\max\{d,\#F(\I_n)\}}
  \end{displaymath}
  which immediately yields the inequality in $(b)$.
\end{proof}

\begin{remark}
  Mora \cite[Thm.~38.2.7]{Mora} presented another version of this
  lemma. Instead of our set $F(\I)$, he defined
  $\tilde{F}(\I)=\{\tau x_n^a \in N(\I) \ | \ \tau \in N(\I_n),\ \deg(\tau
  x_n^a)< \deg{(\I,\prec)}\}$ which differs only in the condition on
  $\tau$. Assuming the equality $\tilde{F}_0(\I)=\tilde{F}(\I_n)$ where
  $\tilde{F}_0(\I)$ contains the elements of $\tilde{F}(\I)$ with $a=0$, he
  proved the following two properties:
  \begin{itemize}
  \item[$(a)$] $\deg(\I,\prec)\le \deg(\I_n,\prec)+\#\tilde{F}(\I_n)$,
  \item[$(b)$] $\#\tilde{F}(\I)\le \bigl(\#\tilde{F}(\I_n)\bigr)^2$.
  \end{itemize}
  However, in general these assertions are not correct -- not even for an
  ideal in generic position. Indeed, in general we have only
  $\tilde{F}(\I_n)\subseteq \tilde{F}_0(\I)$ and if $\dim(\I)>0$ and
  $\deg{(\I_1\prec)}<\deg{(\I,\prec)}$ then equality does not hold. As a
  concrete example consider
  $\I=\langle x_1^2,x_2^{11}x_1\rangle\subset \kk[x_1,x_2]$. We perform a
  generic linear change $x_1=ay_1+by_2$ and $x_2=cy_1+dy_2$ with parameters
  $a,b,c,d\in\kk$.  The leading term ideal of the new ideal is then
  $\langle y_1^2,y_2^{11}y_1\rangle$. This show that $\I=\gin(\I)$ and
  therefore the original coordinates for $\I$ are already generic. We have
  $\I_2=\langle x_1^2\rangle$, $F(\I_2)=\{1,x_1\}$ and
  $ \deg(\I_2,\prec)=2$. Furthermore, we have
  $\tilde{F}(\I)=\{x_2^{11}\}\cup \{x_2^i,x_2^ix_1 \ | \ i=0,\ldots ,10\}$
  and $\# \tilde{F}(\I)=23$. Thus,
  $12=\deg(\I,\prec) \not \le \deg(\I_n,\prec)+\#\tilde{F}(\I_n)=2+2=4$ and
  $23=\#\tilde{F}(\I)\not\le (\#\tilde{F}(\I_n))^2=4$.
\end{remark}

In the case that $\I$ is a zero-dimensional ideal, we can derive explicit
upper bounds for $\deg(\I,\prec)$ and $\#F(\I)$ using the following
well-known lemma.  We include an elementary proof for the sake of
completeness.
\begin{lemma}\label{lem0dim}
  Let $\I$ be a zero-dimensional ideal. Then
  \begin{itemize}
  \item[$(a)$] $\deg(\I,\prec)\le d_1+\cdots +d_n-n+1$,
  \item[$(b)$] $\#F(\I)\le d_1\cdots d_n$.
  \end{itemize}
\end{lemma}

\begin{proof}
  $(a)$ was already proven in Thm.\ \ref{maincor}. We present now an
  elementary proof for $(b)$. The assumption $\dim({\I})=0$ implies that
  $\#F(\I)=\dim_{\kk}(\P/\I)$ and this dimension is equal to the sum of the
  coefficients of the Hilbert series of $\I$ (which is of course a
  polynomial here). We may assume w.l.o.g.\ that the first $n$ generators
  $f_1,\ldots, f_n$ form a regular sequence (Lem.\ \ref{Sec3:Lem1}). Thus
  the Hilbert series of $\I'=\langle f_1,\ldots, f_n \rangle$ is
  $\HS_{\I'}(t)=\prod_{i=1}^{n}{(1+\cdots +t^{d_{i}-1})}$ and
  $\dim_{\kk}(\P/\I')$ is at most $\HS_{\I'}(1)=d_1\cdots d_n$. We
  obviously have $\dim_{\kk}(\P/\I) \le \dim_{\kk}(\P/\I')$ and this proves
  the assertion.  $\qed$
\end{proof}

We state now the main result of this section.

\begin{theorem}\label{Mainthm1}
  If the ideal $\I$ is in strongly stable position, then
  $\#F(\I)\le d^{(n-D)2^{D}}$ and
  \begin{displaymath}
    \deg(\I,\prec)\le 
    \max{\bigl\{(n-D+1)(d-1)+1,2d^{(n-D)2^{D-1}}\bigr\}}\,.
  \end{displaymath}
\end{theorem}

\begin{proof}
  We proceed by induction over $D=\dim(\I)$. In this proof without loss of
  generality, we may assume that $d\ge 2$. If $D=0$, the assertions follow
  immediately from Lem.~\ref{lem0dim}. For $D>0$, we exploit that then
  $\dim(\I)=\dim(\I_n)+1$ and that we may consider $\I_n$ as an ideal in
  $\kk[x_1,\ldots ,x_{n-1}]$.  Lem.~\ref{Mora} now entails that
  \begin{displaymath}
    \begin{aligned}
      \#F(\I) &\le \max\{d,\#F(\I_n)\}^2\\
%      &\le (d^{(n-1)})^2=d^{(n-1)2}
      &\le \bigl(d^{(n-1-(D-1))2^{D-1}}\bigr)^{2}=d^{(n-D)2^{D}}
    \end{aligned}
  \end{displaymath}
  and thus the first inequality.

  For the second inequality, Thm.~\ref{laz3}, which will be proven in the
  last section, provides the starting point for the induction, as it
  immediately implies our claim for $D\le 1$. For $D\ge 2$, we obviously
  have $(n-1-(D-1)+1)(d-1)+1\le d^{(n-D)2^{D-1}}$ and
  $2d^{(n-1-(D-1))2^{D-2}} \le d^{(n-1-(D-1))2^{D-1}}$.  We can thus
  rewrite the induction hypothesis as
  \begin{displaymath}
    \begin{aligned}
      \deg{(\I_n,\prec)} \le 
          \max\bigl\{&(n-1-(D-1)+1)(d-1)+1, \\
                 &2d^{(n-1-(D-1))2^{D-2}}\bigr\} \le d^{(n-D)2^{D-1}}\,.
    \end{aligned}
  \end{displaymath}
  Again by Lem.~\ref{Mora}, we can also estimate
  \begin{displaymath}
    \begin{aligned}
      \deg(\I,\prec) &\le \max{\{d,\deg(\I_n,\prec)\}}+\#F(\I_n)\\
 %     &\le 2d^{(n-1-(D-1))2^{D-2}}+d^{(n-1-(D-1))2^{D-1}}\\
      &\le  d^{(n-D)2^{D-1}}+d^{(n-1-(D-1))2^{D-1}}\\
      &= 2d^{(n-D)2^{D-1}}
    \end{aligned}
  \end{displaymath}
  proving the second assertion. $\qed$
\end{proof}

\begin{example}
  Let us consider the values $n=2,d=2$ and $D=0$. The above theorem states
  $\deg(\I,\prec)\le 2^2=4$.  Consider the ideal
  $\I=\langle x_1^2,x_1x_2+x_2^2 \rangle$.  By performing a generic linear
  change of coordinates, we get
  $\gin(\I)=\langle x_2x_1, x_1^2, x_2^3\rangle$. Therefore
  $\#F(\I)=4\le 4$ and $\deg{(\I,\prec)}=3\le 4$ confirming the accuracy of
  the presented upper bounds. It should be noted that for such a
  zero-dimensional ideal Theorem \ref{maincor} provides the best upper
  bound for $\deg{(\I,\prec)}$, namely $d_1+\cdots +d_n-n+1$ which is equal
  to the exact value $3$ for this example.
\end{example}

Using Prop.~\ref{propdepth}, we obtain even sharper bounds depending on
both the dimension and the depth of $\I$.  We continue to write
$\dim(\I)=D$ and $\depth(\I)=\lambda$.  It is well-known that we always
have $D\ge\lambda$ (a simple proof using Pommaret bases can be found in
\cite{wms:comb2} after Prop.~3.19). If $D=\lambda$, then $\R$ is {\em
  Cohen-Macaulay}. In this case, a nearly optimal upper bound for
$\deg(\I,\prec)$ exists.  Recall that a homogeneous ideal $\I\subset \P$ is
in {\em N\oe ther position}, if the ring extension
$\kk[x_{n-D+1} ,\ldots ,x_{n}] \hookrightarrow \P/\I$ is integral.
Alternatively, N\oe ther position can be defined combinatorially as a
weakened version of quasi stable position (see \cite[Thm.~4.4]{hashemi3}).

\begin{theorem}(\cite[Prop.~4.8, page 117]{Monique}) Let $\R$ be a
  Cohen-Macaulay ring with $\I$ in N\oe ther position. Then,
  $\deg(\I,\prec)\le d_1+\cdots +d_{n-D}-(n-D)+1$.
\end{theorem}

For the rest of this section, we thus assume that $\R$ is not
Cohen-Macaulay, i.\,e.\ that $D>\lambda$.

\begin{corollary}\label{Cor2}
  If $\I$ is in strongly stable position and $D>1$, then
  $\#F(\I)\le d^{(n-D)2^{D-\lambda-1}}$ and
  $\deg(\I,\prec)\le 2d^{(n-D)2^{D-\lambda-1}}$.
\end{corollary}

The maximal degree of an element of the Pommaret basis of an ideal in quasi
stable position equals the Castelnuovo-Mumford regularity
\cite[Cor.~9.5]{wms:comb2}.  If the ideal is even in stable position, then
the Pommaret basis coincides with the reduced Gr\"obner basis
\cite[Thm.~2.15]{Mall}.  These considerations imply now immediately the
following two results.

\begin{corollary}
  If the ideal $\I$ is in strongly stable position and $D> 1$, then
  $\reg(\I)\le 2d^{(n-D)2^{D-\lambda-1}}$.
\end{corollary}

\begin{corollary}
  Let the ideal $\I$ be in quasi stable position, $\H$ its Pommaret basis
  and $D> 1$. If we write $\deg(\H)$ for the maximal degree of an element
  of $\H$, then
  $\deg{(\I,\prec)}\le \deg{(\H)} \le 2d^{(n-D)2^{D-\lambda-1}}$.
\end{corollary}

%%%%%%%%%%%%%%%%%%%%%%%%%
\section{Improving the upper bound of\\ Caviglia-Sbarra}
\label{sec:5}

In 2005, Caviglia and Sbarra \cite{Caviglia} gave a simple proof for the
upper bound $(2d)^{2^{n-2}}$ for $\deg(\I,\prec)$ when the coordinates are
in generic position by analyzing Giusti's proof and exploiting some
properties of quasi stable ideals. We will now improve this bound to a
dimension dependent bound. As a by-product, we will show that the notion of
genericity that one needs here is strongly stable position.

We begin with a quick review of the approach of Caviglia and Sbarra
\cite{Caviglia}.  For any monomial ideal $\J\subset\P$ let $G(\J)$ be its
unique minimal generating set.  We write
$\deg_i(\J)=\max\{\deg_i(u)\ | \ u\in G(\J)\}$ where $\deg_{i}$ denotes the
degree in the variable $x_{i}$. Slightly changing our previous notation, we
now denote by $\J_i$ the ideal
$\J|_{x_{i+1}=\cdots =x_n=0}\subset \kk[x_{1},\ldots ,x_{i}]$. It follows
immediately from the definition of a quasi stable ideal that
$\deg_{i}(\J_i)=\deg_{i}(\J)$. We note that two distinct terms in
$G(\J)$ must differ already in the first $n-1$ variables because of the
minimality of $G(\J)$. Hence $\# G(\J)\le \prod_{i=1}^{n-1}(\deg_i(\J)+1)$.

Assume that $\I$ is in quasi stable position and $\I$ satisfies CP w.r.t.\
$d$. CP implies that $\deg(\I,\prec)-d+1\le \#G(\LT(\I))$ and hence
$\deg(\I,\prec)\le d-1+\prod_{i=1}^{n-1}(\deg_i(\LT(\I))+1)$.  Quasi
stability of $\LT(\I)$ implies that $\deg_i(\LT(\I))=\deg(\I_i,\prec)$ and
thereby $\deg(\I,\prec)\le d-1+\prod_{i=1}^{n-1}(\deg(\I_i,\prec)+1)$. 

Set $B_1=d$ and for $i\ge 2$ recursively
$B_i=d-1+\prod_{j=1}^{i-1}(B_j+1)$. If we assume that for each index
$1\leq i< n$ the reduced ideal $\I_i$ satisfies CP w.r.t.\ $d$, then by the
considerations above $\deg(\I_i,\prec)\le B_i$.  In particular, $B_2=2d$
and $\deg(\I,\prec)\le B_n$.  One easily sees that the $B_{i}$ satisfy the
recursion relation
$B_i=d-1+(B_{i-1}+1)(B_{i-1}-d+1)=B_{i-1}^2-(d-2)B_{i-1}$ for all $i\ge
2$. Since we may suppose that $d\ge 2$, we have $B_{i}\le B_{i-1}^2$. Thus,
for all $i\ge 2$ we have $B_i\le (2d)^{2^{i-2}}$ and therefore
$B_{n}=\deg(\I,\prec)\le (2d)^{2^{n-2}}$. We summarize the above discussion
in the next theorem.

\begin{theorem}(\cite{Caviglia})\label{Mainthm2} 
  Suppose that $\I$ is in quasi stable position and that the ideals
  $\I_1,\ldots, \I_{n-1},\I$ satisfy CP w.r.t.~$d$. Then
  $\deg(\I,\prec)\le \reg(\I)\le (2d)^{2^{n-2}}$.
\end{theorem}

\begin{proof}
  We mentioned already above that for any ideal in quasi stable position
  $\deg(\I,\prec)\le \reg(\I)$, since the regularity equals the maximal
  degree of an element of the Pommaret basis of $\I$.  As the regularity
  remains invariant under linear coordinate transformations, we may
  w.l.o.g.\ assume that $\I$ is even in strongly stable position where
  $\deg(\I,\prec)=\reg(\I)$ and where Prop.~\ref{popstrong} entails that
  also $\I_1,\ldots, \I_{n-1},\I$ satisfy CP w.r.t.~$d$.  Now the assertion
  follows from the consideration above.
\end{proof}

We derive now a dimension dependent upper bound for $\deg(\I,\prec)$.

\begin{theorem}\label{Mainthm3}
  Suppose that $\I$ is in strongly stable position and $D=\dim(\I)\ge 1$.
  Then
  \begin{displaymath}
    \deg(\I,\prec)=\reg(\I)\le \bigl(d^{n-D}+(n-D)(d-1)\bigr)^{2^{D-1}}\,.
  \end{displaymath}
\end{theorem}

\begin{proof}
  Since $\I$ is in strongly stable position, the ideal
  $\I_{n-D}\subset \kk[x_1,\ldots ,x_{n-D}]$ is zero-dimensional
  \cite[Prop~3.15]{wms:comb2}. According to Lem. \ref{lem0dim},
  $\deg(\I_{n-D})\le (n-D)(d-1)+1$. Hence the maximal degree of a term in
  $G(\LT(\I))$ which depends only on $x_1,\ldots ,x_{n-D}$ is at most this
  bound. We shall now construct an upper bound for the degree of the terms
  in $G(\LT(\I))$ containing at least one of the remaining variables
  $x_{n-D+1},\ldots ,x_n$.  Following the approach of Caviglia and Sbarra,
  we first look for an upper bound for the number of these terms.

  Consider a term $m=x_1^{\alpha_1} \cdots x_n^{\alpha_n}\in G(\LT(\I))$
  with $\alpha_i>0$ for some $i\ge n-D+1$. It is clear that
  $x_1^{\alpha_1} \cdots x_{n-D}^{\alpha_{n-D}}$ belongs to the complement
  of $\LT(\I_{n-D})$. Since the ideal
  $\I_{n-D}\subset \kk[x_1,\ldots ,x_{n-D}]$ is zero-dimensional,
  Lem.~\ref{lem0dim} entails that
  $\dim_{\kk}\bigl(\kk[x_1,\ldots ,x_{n-D}]/\I_{n-D}\bigr)\le
  d^{n-D}$. Hence the number of terms
  $x_1^{\alpha_1} \cdots x_{n-D}^{\alpha_{n-D}}$ is at most $d^{n-D}$. On
  the other hand, for any index $n-D+1\le i\le n$ we have
  $\alpha_i\le\deg_{i}(\LT(\I))\le \deg(\I_i,\prec)$. Furthermore, we know
  that two distinct term in $G(\LT(\I))$ differ already in their first
  $n-1$ variables. These arguments imply that the number of terms in
  $G(\LT(\I))$ containing at least one of the variables
  $x_{n-D+1},\ldots ,x_n$ is at most
  $d^{n-D}\prod_{i=n-D+1}^{n-1}\bigl(\deg(\I_i,\prec)+1\bigr)$.

  The strongly stability of $\I$ implies that CP holds for $\LT(\I)$
  w.r.t.\ $(n-D)(d-1)+1\ge d$ by Prop.~\ref{popstrong}. Hence
  $\deg{(\I,\prec)}-\bigl((n-D)(d-1)+1\bigr)+1$ must be less than or equal
  to the number of terms in $G(\LT(\I))$ containing at least one of the
  variables $x_{n-D+1},\ldots ,x_n$ leading to the estimate
  \begin{displaymath}
    \deg(\I,\prec)\le 
    d^{n-D}\prod_{i=n-D+1}^{n-1}\bigl(\deg(\I_i,\prec)+1\bigr)+(n-D)(d-1)\,.
  \end{displaymath}
  Set $B_{n-D+1}=d^{n-D}+(n-D)(d-1)$ and recursively
  $B_j=d^{n-D}\prod_{i=n-D+1}^{j-1}(B_i+1)+(n-D)(d-1)$ for
  $n-D+2\le j\le n$. One easily verifies that these numbers satisfy the
  recursion relation
  $B_j= \bigl(B_{j-1}-(n-D)(d-1)\bigr)(B_{j-1}+1)+(n-D)(d-1)=
  B_{j-1}^2-\bigl((n-D)(d-1)-1\bigr)B_{j-1}$. We may again assume that
  $d\ge 2$, and therefore $B_j\le B_{j-1}^2$ for $n-D+2\le j\le n$. This
  implies that $B_{j}\le (d^{n-D}+(n-D)(d-1))^{2^{j-n+D-1}}$ and in
  particular we have $B_{n}\le (d^{n-D}+(n-D)(d-1))^{2^{D-1}}$. $\qed$
\end{proof}

\begin{remark}
  Let us compare the dimension dependent bounds
  $A(n,d,D)=2d^{(n-D)2^{D-1}}$ derived in Thm.~\ref{Mainthm1} and
  $B(n,d,D)=2(1/2d^{n-D}+d)^{2^{D-1}}$ due to Mayr and Ritscher \cite{MR}
  with $C(n,d,D)=(d^{n-D}+(n-D)(d-1))^{2^{D-1}}$ obtained now.  Obviously,
  all three bounds describe essentially the same qualitative behaviour,
  although they are derived with fairly different approaches.  However, the
  bound $B(n,d,D)$ of Mayr and Ritscher has almost always the best
  constants.  But there are some cases where one of the other bounds is
  better.  For example, in the case of a hypersurface, i.e.\ for $D=n-1$,
  $A(n,d,D)$ is smaller than $B(n,d,D)$.  For some curves of low degree,
  i.e.\ for $D=1$ and small values of $d$, $C(n,d,D)$ is smaller than
  $B(n,d,D)$.  Some concrete inequalities are:
  \begin{itemize}
  \item $A(5, 3, 4)<C(5, 3, 4)$,
  \item $A(3, 5, 2)>C(3, 5, 2)$,
  \item $A(5, 2, 4)<B(5, 2, 4)$,
  \item $A(5,4,2)>B(5,4,2)$,
  \item $B(4, 5, 1)>C(4, 5, 1)$,
  \item $B(5, 2, 3)<C(5, 2, 3)$.
  \end{itemize}
  Hence no bound is always the best one.
\end{remark}

Again an application of Prop.~\ref{propdepth} yields immediately an
improved bound depending on both the depth and the dimension of $\I$.

\begin{corollary}\label{Cavag}
  Under the assumptions of Thm.~\ref{Mainthm3}, one has
  $\deg(\I,\prec)=\reg(\I)\le (d^{n-D}+(n-D)(d-1))^{2^{D-\lambda-1}}$.
\end{corollary}

It should be noted that in positive characteristic it is not always
possible to achieve strongly stable position by linear coordinate
transformations (see \cite{hashemi3} for a more detailed
discussion). Nevertheless, following \cite{Caviglia}, we state the
following conjecture.

\begin{conjecture}
  The upper bound for the Castelnuovo-Mumford regularity of $\I$ in
  Cor.~\ref{Cavag} holds independently of the characteristic of $\kk$.
\end{conjecture}

%%%%%%%%%%%%%%%%%%%%%%%%%
\section{ Lazard's upper bound}
\label{sec:6}

Finally, in this section we study Lazard's upper bound \cite{Daniel83} for the
degree of Gr\"obner bases for both  homogeneous and non-homogeneous
ideals.  We provide a simple proof for his results and
generalize Giusti's bound to non-homogeneous ideals. Note that for Lazard
\cite{Daniel83} dimension was always the one as projective variety, whereas
we use throughout this paper the one as affine variety which is one
higher.  In the sequel, we always set $d_{i}=1$ for any $i>k$.

\begin{theorem}(\cite[Thm.~2]{Daniel83})\label{laz1} 
  Assume that $\dim(\I)\le 1$. Then we have $\deg{(\gin(\I),\prec)}\le
  d_1+\cdots +d_r-r+1$ where $r=n-\lambda$. 
\end{theorem}

We showed in \cite{hashemi2} that many properties of $\gin(\I)$ also hold
for $\lt(\I)$ provided $\I$ is in quasi stable position.  Along these
lines, we shall now prove that in Lazard's upper bound we can replace
$\gin(\I)$ by $\I$, if $\I$ is in quasi stable position. For this, we need
the next proposition also due to Lazard, which is the key point in the
proof of the above theorem.

\begin{proposition}[{\cite[Thm.~3.3]{Lazard81}}]\label{laz2}
  Assume again that $\dim(\I)\le 1$.  Then
  $\dim_{\kk}(\P/\I)_\ell=\dim_{\kk}(\P/\I)_{\ell+1}$ for each
  $\ell \ge d_1+\cdots +d_n-n+1$.
\end{proposition}

Thus, under the assumptions of this proposition, we can say that
$\hilb(\I)\le d_1+\cdots +d_n-n+1$.

\begin{theorem}\label{laz3}
  Suppose that $\I$ is in quasi stable position and $\dim(\I)\le 1$. Then,
  $\deg(\I,\prec)\le d_1+\cdots +d_r-r+1$ where $r=n-\lambda$.
\end{theorem}

\begin{proof}
  It suffices to show that $\deg(\I,\prec)\le d_1+\cdots +d_n-n+1$, since
  then the desired inequality follows immediately from
  Prop.~\ref{propdepth}. As $\I$ is in quasi stable position, we have the
  inequality $\deg(\I,\prec)\le \max\{\hilb(\I), \hilb(\I')\}$ where
  $\I'=(\I+\li x_n\ri)\cap \kk[x_1,\ldots ,x_{n-1}]$ is an ideal in the
  ring $\kk[x_1,\ldots ,x_{n-1}]$ \cite[Thm. 4.17]{hashemi1}, \cite[Thm.
  4.7]{Seiler}. Obviously, $\I'$ is generated by
  $f_1|_{x_n=0},\ldots ,f_k|_{x_n=0}$ and $\dim(\I)\le 1$ (by using the fact that $\I$ is in quasi stable position) entails
  $\dim(\I')\le 1$. These arguments show that, by Prop.~\ref{laz2},
  $\hilb(\I) \le d_1+\cdots +d_n-n+1$ and
  $\hilb(\I')\le d_1+\cdots +d_{n-1}-(n-1)+1$ which proves the
  assertion. $\qed$
\end{proof}

\begin{example}
  Lazard \cite[Conj.~3]{Daniel83} conjectured that the conclusion of
  Theorem \ref{laz1} remained true, if one replaces $\gin{\I}$ by
  $\I$. Mora claimed that the following ideal (see the Appendix of
  \cite{Daniel83}) provided a counter-example. Consider the homogeneous
  ideal
  $\I=\li x_1x_2^{t-1}-x_3^t, x_1^{t+1}-x_2x_3^{t-1}x_4,
  x_1^tx_3-x_2^tx_4\ri$ in the polynomial ring
  $\P=\kk[x_1,\ldots,x_4]$. Thus we have $d_{1}=t, d_{2}=d_{3}=t+1$.  One
  can show that the polynomial $x_3^{t^2+1}-x_2^{t^2}x_4$ appears in the
  Gr\"obner basis of $\I$ and hence $\deg(\I,\prec)\ge t^{2}+1$. For
  simplicity we restrict to the case $t=4$ where we obtain
  \begin{displaymath}
    \LT(\I)=\li x_1x_2^3, x_1^4x_3, x_1^5, x_1^3x_3^5, x_1^2x_3^9, 
                x_1x_3^{13}, x_3^{17} \ri\,.
  \end{displaymath}
  Thus we find here $\deg(\I,\prec)=17>d_{1}+d_{2}+d_{3}-3=11$.  But as
  $\dim(\I)=2$, $\I$ does not yield a counter-example to Lazard's
  conjecture. However, if we consider
  $\I'=\I|_{x_4=0}\subset \kk[x_1,x_2,x_3]$, then we find that $\I'$ has
  dimension $1$ and that $\LT(\I')$ is generated by the same terms as
  $\LT(\I)$.  $\I'$ is not in quasi stable position, as no pure power of
  $x_2$ belongs to $\LT(\I')$.  Hence $\I'$ represents a counter-example to
  Lazard's conjecture.  This example shows furthermore that in Thm.~\ref{laz3} it is not
  possible to drop the assumption of quasi stable position.
\end{example}

\begin{remark}
  We gave above a direct proof for Thm.~\ref{laz3}. However, we can provide
  a more concise proof using Thm.~\ref{laz1} and Pommaret bases. Indeed,
  from Thm.~\ref{laz1} it follows that $\reg(\I)\le d_1+\cdots +d_r-r+1$
  where $r=n-\lambda$, as $\reg(\I)=\deg(\gin(\I),\prec)$. Since the ideal
  $\I$ is in quasi stable position, it possesses a finite Pommaret basis
  $\H$ where $\reg(\I)$ is the maximal degree of the elements of $\H$ and therefore
  $\deg(\I,\prec)\le d_1+\cdots +d_r-r+1$.  These considerations also yield
  immediately the following corollary.
\end{remark}

\begin{corollary}\label{Laz4}
  If $\dim(\I)\le 1$, then $\reg(\I)\le d_1+\cdots +d_r-r+1$ where $r=n-\lambda$.   
\end{corollary}

Finally, we present an affine version of Thm.~\ref{laz3}. We drop now the
assumption that the polynomials $f_1,\ldots ,f_k$ generating $\I$ are
homogeneous. Let $x_{n+1}$ be an extra variable and $\tilde{f}$ the
homogenization of $f$ using $x_{n+1}$. We further denote by $\tilde{\I}$
the ideal generated by $\tilde{f_1}, \ldots ,\tilde{f_k}$ (note that in
general this is \emph{not} equal to the homogenization of $\I$). The next
proposition may be considered as a generalization of Lazard's upper bound
\cite[Thm.~2]{Daniel83} to ideals in quasi stable position.

\begin{proposition}\label{Laz5}
  Assume that $\tilde{\I}$ is in quasi stable position, that
  $\dim(\tilde{\I})\le 1$ and that $\depth(\tilde{\I})=\lambda$.  Then,
  $\deg(\I,\prec) \le d_1+\cdots +d_r-r+1$ where $r=n+1-\lambda$.
\end{proposition}

\begin{proof}
  By Thm.~\ref{laz3},
  $\hilb(\tilde{\I})\le \reg(\tilde{\I}) \le d_1+\cdots +d_r-r+1$ where
  $r=n+1-\lambda$. Hence
  $\dim_\kk(\kk[x_1,\ldots
  ,x_{n+1}]/\tilde{\I})_{\ell}=\dim_\kk(\kk[x_1,\ldots
  ,x_{n+1}]/\tilde{\I})_{\ell+1}$ for all degrees
  $\ell\ge d_1+\cdots +d_r-r+1$. Therefore, we have
  $\dim_\kk(\P/\I)_{\le \ell}=\dim_\kk(\P/\I)_{\le \ell+1}$ for each
  $\ell\ge d_1+\cdots +d_r-r+1$ and this observation implies that the
  reduced Gr\"obner basis of $\I$ contains no element of degree greater
  than $d_1+\cdots +d_r-r+1$. $\qed$
\end{proof}

We conclude this paper by mentioning that it is easy to see that for a
homogeneous ideal with $\dim(\I)\le 1$, being in quasi stable position is
equivalent to being in N\oe ther position. This implies that in
Thm.~\ref{laz3} and Prop.~\ref{Laz5} one can replace ``quasi stable
position'' by ``N\oe ther position''.

\section*{Acknowledgments.}  

The research of the first author was in part supported by a grant from IPM
(No. 94550420).  The work of the second author was partially performed as
part of the H2020-FETOPEN-2016-2017-CSA project $SC^{2}$ (712689). The
authors would like to thank the anonymous reviewers for their valuable
comments.

\bibliographystyle{plain}

\begin{thebibliography}{10}

  \bibitem{Becker} { T. Becker and V. Weispfenning}. 
\newblock {\it Gr\"obner Bases, a Computational Approach to Commutative Algebra}. 
\newblock Springer, 1993.

\bibitem{bayer}
D.~Bayer.
\newblock {\em The Division Algorithm and the Hilbert Scheme}.
\newblock Ph.D. thesis, Harvard University, 1982.

\bibitem{bayer_stillman}
D. Bayer and M. Stillman.
\newblock A Criterion for Detecting {$m$}-Regularity. 
\newblock {\em Invent. Math.} {\bf{87}}(1), pages 1--11, 1987.

\bibitem{Bermejo2}
I. Bermejo and P. Gimenez.
\newblock Computing the {C}astelnuovo-{M}umford Regularity of Some Subschemes of {${\mathbb P}\sb K\sp n$} Using Quotients of Monomial Ideals.
\newblock {\em J. Pure Appl. Algebra}, {\bf 164}(1-2), pages 23--33, 2001.

\bibitem{bermejo}
I. Bermejo and Ph. Gimenez.
\newblock Saturation and {C}atelnuovo-{M}umford regularity.
\newblock {\em J. Algebra}, {\bf 303}, pages 592--617, 2006.


\bibitem{Bruno1}
B. Buchberger.
\newblock {\em Ein Algorithmus zum Auffinden der Basiselemente des Restklassenringes nach einem nulldimensionalen Polynomideal.}
\newblock PhD thesis, Universit\"{a}t Innsbruck, 1965.

\bibitem{Bruno2}
B. Buchberger.
\newblock A Criterion for Detecting Unnecessary Reductions in the Construction of Gr\"obner Bases.
 \newblock  {\em EUROSAM'79,  Lecture Notes in Compute. Sci., Springer}, {\bf 72}, pages 3--21, 1979.

\bibitem{Caviglia}
G. Caviglia and E. Sbarra. 
\newblock Characteristic-free Bounds for the Castelnuovo-Mumford Regularity. 
\newblock {\em Compos. Math.},  {\bf{141}}(6), pages 1365--1373, 2005.

\bibitem{little}
D. Cox, J. Little and D.  O'Shea.
\newblock {\em Ideals, Varieties, and Algorithms.} 
Springer-Verlag, third edition, 2007. 

\bibitem{Dube}
T. Dub\'e.
\newblock The Structure of Polynomial Ideals and Gr\"obner Bases.
\newblock {\em SIAM Journal on Computing}, {\bf 19}, pages 750--773, 1990.

\bibitem{Eisenbud_book}
D. Eisenbud.
\newblock {\em Commutative Algebra with a View toward Algebraic Geometry.}
\newblock Springer-Verlag,  1995.


\bibitem{Eisenbud_Goto}
D. Eisenbud and S. Goto. 
\newblock Linear Free Resolutions and Minimal Multiplicity. 
\newblock {\em J. Algebra},  {\bf{88}}(1), pages 89--133, 1984.

\bibitem{Ralf}
R. Fr\"oberg.
\newblock {\em An Introduction to Gr\"obner Bases.} 
John Wiley \& Sons Ltd., Chichester, 1997. 

\bibitem{Gal1}
A. Galligo. 
\newblock A propos du th\'eor\`eme de preparation de Weierstrass. 
\newblock {\em Lect. Notes Math.},  {\bf 409}, pages 543--579, 1974. 

\bibitem{Gal2}
A. Galligo. 
\newblock Th\'eor\`eme de division et stabilit\'e en g\'eom\'etrie analytique locale.. 
\newblock {\em Ann. Inst. Fourier}, {\bf 29}(2), pages 107--184, 1979.

\bibitem{Giusti1}
M.~Giusti.
\newblock Some Effectivity Problems in Polynomial Ideal Theory.
\newblock {\em EUROSAM'84, Lecture Notes in Comput. Sci.}, {\bf{174}},
pages 159--171, Springer, 1984. 

\bibitem{Green}
  M. Green.
  \newblock  Generic Initial Ideals.
  \newblock In: Elias, J., Giral, J., Mir{\'o}-Roig, R.,
  Zarzuela, S. (eds.) Six Lectures on Commutative Algebra, pages 119--186.
  {\em Progress in Mathematics} {\bf 166}, Birkh\"auser, Basel, 1998.


\bibitem{Singular}
G.M. Greuel and G. Pfister.
\newblock {\em A {\bf {S}ingular} Introduction to Commutative Algebra}.
\newblock Springer-Verlag, 2008.


\bibitem{hashemi1}
  A. Hashemi. Strong Noether Position and Stabilized Regularities.
  {\em Commun. Algebra}, {\bf 38}(2), pages 515--533, 2010.

\bibitem{hashemi2}
A. Hashemi, M. Schweinfurter and W.M. Seiler. Quasi-stability versus
Genericity. {\em Proc. CASC'12, Lecture Notes in Comput. Sci.}, {\bf 7442},
pages 172--184, 2012. 

\bibitem{hashemi3}
A. Hashemi, M. Schweinfurter and W.M. Seiler. Deterministic Genericity
for Polynomial Ideals. {\em J. Symb. Comput.}, accepted for publication
(2017), 31 pages ({\tt doi.org/10.1016/j.jsc.2017.03.008}).  

\bibitem{Lazard81}
D. Lazard. 
\newblock R\'esolution des syst\`emes d'\'equations alg\'ebriques. 
\newblock {\em Theoret. Comput. Sci.} {\bf{15}}(1), pages 77--110, 1981.

\bibitem{Daniel83}
D.~Lazard.
\newblock Gr\"obner Bases, {G}aussian Elimination and Resolution of Systems of
  Algebraic Equations.
\newblock  {\em Proc. EUROCAL'83, Lecture  Notes in Comput. Sci.}, {\bf
  162}, pages 146--156, Springer, 1983. 

\bibitem{Lazard3}
D.~Lazard.
\newblock A Note on Upper Bounds for Ideal-Theoretic Problems. 
\newblock  {\em J. Symb. Comput.}, {\bf 13}(3), pages   231--234, 1992. 

\bibitem{Monique}
M. Lejeune-Jalabert.
\newblock{\em Effectivit\'e de calculs polynomiaux.}
\newblock Cours de D.E.A, Institute Fourier, Grenoble. 1984.

\bibitem{Mall}
D. Mall.
\newblock On the Relation between Gr\"obner and Pommaret Bases.
\newblock{\em Appl. Alg. Eng. Comm. Comp.}, {\bf 9}, pages 117--123, 1998.

\bibitem{Mat}
  H. Matsumura.
 \newblock {\em Commutative ring theory.}  Cambridge University Press, second edition, 1989.

\bibitem{MM}
E. W. Mayr and A. R. Meyer.
\newblock The Complexity of the Word Problems for Commutative Semigroups and
  Polynomial Ideals.
\newblock {\em Adv. Math.}, {\bf 46}(3), pages 305--329, 1982.

\bibitem{MR}
E.W. Mayr and S. Ritscher.
\newblock Dimension-Dependent Bounds for Gr\"obner Bases of Polynomial Ideals.
\newblock  {\em J. Symb. Comput.}, {\bf 49}, pages  78--94, 2013. 

\bibitem{Moller}
H.M. M{\"o}ller and F. Mora.
\newblock Upper and Lower Bounds for the Degree of {G}roebner Bases.
\newblock In {\em Proc. EUROSAM'84, Lecture
  Notes in Comput. Sci.}, {\bf 174}, pages 172--183. Springer, 1984.

\bibitem{Mora}
T. Mora.
\newblock {\em Solving Polynomial Equation Systems II: Macaulay's Paradigm
  and Gr\"obner Technology}. 
\newblock Cambridge University Press, 2005.

\bibitem{Mumford}
D. Mumford. 
\newblock{\em Lectures on Curves on an Algebraic Surface}. 
Princeton University Press, 1966.

\bibitem{wms:comb1}
W.M. Seiler. A Combinatorial Approach to Involution and $\delta$-Regularity {I}:
  Involutive Bases in Polynomial Algebras of Solvable Type. {\em Appl.\ Alg.\ Eng.\
  Comm.\ Comp.} {\bf  20}, pages  207--259, 2009.

\bibitem{wms:comb2}
W.M. Seiler. A Combinatorial Approach to Involution and $\delta$-Regularity
  {II}: Structure Analysis of Polynomial Modules with {P}ommaret Bases. {\em Appl.\
  Alg.\ Eng.\ Comm.\ Comp.}  {\bf 20}, pages  261--338, 2009.

\bibitem{wms:invol}
W.M. Seiler. {\em Involution -- {T}he Formal Theory of Differential
  Equations and its Applications in Computer Algebra.} Springer-Verlag, 2010.

\bibitem{Seiler}
 W.M. Seiler.  Effective Genericity, $\delta$-Regularity and Strong Noether
 Position. {\em Commun. Algebra}, {\bf 40}(10), pages 3933--3949, 2012. 
 
\bibitem{Yap}
C.~Yap.
\newblock A New Lower Bound Construction for the Word Problem for
Commutative Thue Systems.  
\newblock {\em J. Symb. Comput.}, {\bf 12}(1), pages  1--28, 1991. 

\end{thebibliography}

\end{document}